\documentclass[journal]{IEEEtran}%[11pt,draftcls,onecolumn]{IEEEtran}%[journal]{IEEEtran}%[12pt,draftcls,onecolumn]
\usepackage{amssymb}
\usepackage{amsmath}
\usepackage{cite}
\usepackage{algorithm}
\usepackage{algorithmic}
\usepackage[dvips]{graphicx}
\usepackage{stfloats}
\begin{document}
%
% paper title
\title{Relay Beamforming Design with SIC Detection for MIMO Multi-Relay Networks with Imperfect CSI}

\author{Zijian Wang, and Wen Chen,~\IEEEmembership{Senior Member,~IEEE}
\thanks{Copyright (c) 2013 IEEE. Personal use of this material is permitted. However, permission to use this material for any other purposes must be obtained from the IEEE by sending a request to pubs-permissions@ieee.org. }
\thanks{The authors are with Department of Electronic Engineering, Shanghai
        Jiao Tong University, China. e-mail: \{wangzijian1786; wenchen\}@sjtu.edu.cn.}
\thanks{This work is supported by the National 973 Project \#2012CB316106 and
\#2009CB824904, by NSF China \#60972031 and \#61161130529.}% <-this % stops a space
}

% The paper headers
%\markboth{IEEE Transactions on Vehicular Technology.}{}
\maketitle

\begin{abstract}
In this paper, we consider a dual-hop Multiple Input Multiple Output
(MIMO) wireless multi-relay network, in which a source-destination pair
both equipped with multiple antennas communicates through multiple
half-duplex amplify-and-forward (AF)  relay terminals which are also with multiple antennas. Since perfect channel state information (CSI) is difficult to obtain in practical multi-relay network, we consider imperfect CSI for all channels. We focus on maximizing the signal-to-interference-plus-noise ratio (SINR) at the destination. We
propose a novel robust linear beamforming at the relays, based on the QR decomposition filter at the destination node which
performs successive interference cancellation (SIC). Using Law of
Large Number, we obtain the asymptotic rate in
the presence of imperfect CSI, upon which, the proposed relay beamforming is optimized.
 Simulation results
show that the asymptotic rate matches with the  ergodic rate
 well. Analysis and simulation results demonstrate that the
proposed beamforming outperforms the conventional beamforming schemes for any power
of CSI errors and SNR regions.
%we found that MF beamformer outperforms MF-RZF when CSI error is considerable. By using the Taylor expansion, we generalize the impact of CSI error to the capacity as an additional CEG-noise in the denominator in the SNR. We then derive the closed form of ergodic capacities for large number of relay nodes ($K$). Utilizing the closed form, we optimize the MF-RZF beamformer.
%Simulation results demonstrate that the closed form for large $K$ is tight and the optimized MF-RZF outperforms other conditional MF-RZF.
\end{abstract}

\begin{keywords}
MIMO relay, successive interference cancellation (SIC) detection, relay beamforming, channel state information (CSI), rate.
\end{keywords}

%\newpage
\section{Introduction}
Relay communications can extend the coverage of wireless networks
and improve spatial diversity of cooperative systems~\cite{rev1}.
 Meanwhile, MIMO
technique is well verified to provide significant improvement in the
spectral efficiency and link reliability because of the multiplexing
and diversity gains~\cite{1,2}. Combining the relaying and MIMO
techniques can make use of  both advantages to increase the data
rate in the cellular edge and extend the network coverage~\cite{rev2}.

MIMO relay networks and MIMO broadcasting relay networks have been extensively investigated in~\cite{3,5,7,31,32,33,rev3}. In addition MIMO multi-relay networks have been studied  in~\cite{4,6,8,9,SVD,iterative,10}. In~\cite{4}, the authors show that the
corresponding  network capacity scales as $C=(M/2) \log(K)+O(1)$,
where $M$ is the number of antennas at the source and
$K\rightarrow\infty$ is the number of relays. The authors also propose a
simple protocol to achieve the upper bound as $K\rightarrow\infty$ when perfect channel state informations (CSIs) of both backward channels (BC) and forward channels (FC) are available at the relay nodes. When CSIs are not available at the relays, a simple AF beamforming protocol is proposed at the relays, but the distributed array gain is not obtained.
%In~\cite{6}, a% linear relaying scheme based on minimum mean square error (MMSE) fulfilling the target SNRs on
%different substreams is proposed and the power-efficient relaying
%strategy is derived in closed form for a MIMO multi-relay network.
In~\cite{8,9}, the authors design three relay beamforming schemes
based on matrix triangularization which have superiority over the
conventional zero-forcing (ZF) and amplify-and-forward (AF)
beamformers. The proposed beamforming scheme can both  fulfill
intranode gain and distributed array gain. In \cite{SVD}, the
authors design a beamforming scheme that achieves the upper bound of
capacity with a small gap when $K$ is significantly large. But it
has bad performance for small $K$ and the source needs CSI which
increases overhead. A unified algorithm is proposed in \cite{iterative} for the optimal linear transceivers at the source
and relays for both one-way and two-way networks. In~\cite{10}, two efficient relay-beamformers for the dual-hop MIMO
multi-relay networks have been presented, which are
based on matched filter (MF) and regularized zero-forcing (RZF), and
utilize QR decomposition (QRD) of the effective system channel
matrix at the destination node~\cite{13}. The beamformers at the
relay nodes can exploit the distributed array gain by diagonalizing
both the backward and forward
channels. The QRD can exploit the intranode array gain by successive interference cancellation (SIC)
detection. These two beamforming schemes  have lower complexity because they
only need one QR decomposition at destination.

On the other hand, all the works for multi-relay MIMO system only
consider  perfect CSI to design
beamformers at the relays or successive interference cancelation
 matrices at the destination. For the multi-relay networks,
imperfect CSI is a practical consideration~\cite{4}. Especially,
knowledge for the CSI of FC at relays will result in large delay and
significant training overhead, because the CSI of FCs
at relays are obtained through feedback links to multiple
relays~\cite{40}. The imperfect CSI of BC at relays is also
practical because of channel estimation error.

For the works on imperfect CSI,
the ergodic capacity and BER performance of MIMO with imperfect CSI is considered in~\cite{14,16,19}.
In~\cite{14}, the authors investigated lower and upper bounds
of mutual information under CSI error. In~\cite{16}, the authors
studied BER performance of MIMO system under combined beamforming
and maximal ratio combining (MRC) with imperfect CSI. In~\cite{19},
bit error probability (BEP) is analyzed based on Taylor
approximation.
Some optimization problem has been investigated with
imperfect CSI in~\cite{15,20,21}. In~\cite{15}, the authors
maximize a lower bound of capacity by optimally configuring the
number of antennas with imperfect CSI.
%In~\cite{18}, assuming only
%imperfect CSI at the relay, optimization problem of maximizing upper
%bound of mutual information is presented and solved.
In~\cite{21}, the authors studied
the trade-off between accuracy of channel estimation and data
transmission, and show that the optimal number of training symbols
is equal to the number of transmit antennas. In~\cite{34}, the authors investigate the effects of channel estimation error on the receiver of MIMO AF two-way relaying networks.

%Inspired by the works on  imperfect CSI and~\cite{10},
In this
paper, we propose a new robust beamforming schemes for dual-hop MIMO
multi-relay networks under the condition of imperfect CSI. SIC is
also implemented at the destination by QR decomposition. The
proposed beamformer at relay is based on the minimum mean square error (MMSE) receiver and the RZF
precoder. We focus on optimizing the regularizing factors in them.
We first optimize the factor $\alpha^{\mathrm{MMSE}}$ in MMSE, and
then optimize the factor $\alpha^{\mathrm{RZF}}$ in RZF for a given
$\alpha^{\mathrm{MMSE}}$. In the derivation, using Law of
Large Number, we obtain the asymtotic rate capacity for the
MMSE-RZF beamformer, based on which,  the performance of the
beamformer for imperfect CSI can be easily analyzed. Simulation
results show that the asymptotic rate capacity matches with the
ergodic capacity well. The asymptotic rate also validates the scaling law in~\cite{4},
when the imperfect CSI presents. Analysis and simulations demonstrate
that the rate of MMSE-RZF outperforms other schemes whether CSI
is perfect or not.
The ceiling effect of the rate capacity and the situation that
CSI error increases with the number of relays are also discussed in
this paper.

The remainder of this paper is organized as follows.  In Section
\uppercase \expandafter {\romannumeral 2}, the system model of a
dual-hop MIMO multi-relay network is introduced. In Section \uppercase
\expandafter {\romannumeral 3}, we explain the MMSE-RZF based beamforming
scheme and QR decomposition. In Section \uppercase \expandafter
{\romannumeral 4}, we optimize the MMSE-RZF and obtain the asymptotic rate of the system. Section \uppercase \expandafter
{\romannumeral 5} devotes to simulation results followed by
conclusion in Section \uppercase \expandafter {\romannumeral 6}.

%The remainder of this paper is organized as follows. In Section
%\uppercase \expandafter {\romannumeral 2}, the system and signal
%model  of the dual-hop MIMO-Relay networks is introduced. In Section
%\uppercase \expandafter {\romannumeral 3}, several beamforming
%design schemes on the relay nodes are presented, followed by the
%simulation results in Section \uppercase \expandafter {\romannumeral
%4}. Finally, conclusions are drawn in Section \uppercase
%\expandafter {\romannumeral 5}.

In this paper, boldface lowercase letter and boldface uppercase
letter represent vectors and matrices, respectively. Notations $\left( {\bf{A}} \right)_{i}$ and  $\left( {\bf{A}} \right)_{i,j}$
denote the $i$-th row and $(i,j)$-th entry of the matrix
${\bf{A}}$. Notations $\mathrm{tr}(\cdot)$, $(\cdot)^{\dag}$, $(\cdot)^{*}$ and $(\cdot)^H$ denote
trace, pseudo-inverse, conjugate and conjugate transpose operation of a matrix respectively. Term
$\mbox{\boldmath $\mathbf{I}$}_N$ is an $N{\times}N$ identity
matrix. The $\mathrm{diag}\left\{\{a_m\}_{m=1}^M\right\}$ denotes a diagonal matrix with diagonal entries of $a_1,\ldots,a_M$. $\|\mathbf{a}\|$ stands for the Euclidean norm of a vector $\mathbf{a}$, and  $\overset{w.p.}{\longrightarrow}$ represents convergence with probability one. Finally, we denote the expectation operation by $\mathrm{E}\left\{\cdot\right\}$.

\section{System Model}
The considered MIMO multi-relay network consists of a single source
and destination node both equipped with $M$ antennas, and $K$
$N$-antenna relay nodes distributed between the source-destination
pair as illustrated in Fig.~1. When the source node implements
spatial multiplexing, the requirement $N \ge M$ must be satisfied if
each relay node is supposed to support all the $M$ independent data
streams. We assume $M=N$ in this paper, while the proposed beamforming scheme and the results can be easily expanded to the case $N > M$. We consider half-duplex non-regenerative relaying
throughout this paper, where it takes two non-overlapping time slots
for the data to be transmitted from the source to the destination
node via the backward channels and forward channels. Due to deep
large-scale fading effects produced by the long distance, we assume
that there is no direct link between the source and destination.  In
a practical system, each relay needs to transmit training sequences
or pilots to acquire the CSI of all channels. Imperfect
channel estimation and limited feedback are also practical
considerations. So in this paper, imperfect CSIs of BC and FC are
assumed to be available at relay nodes. Assume that $\widehat{\bf{H}}_k\in \mathbb{C}^{M\times
M}$ and $\widehat{\bf{G}}_k\in \mathbb{C}^{M\times M}$ stands for
the available imperfect CSIs of BC and FC at the $k$-th relay. We
model the CSIs of BC and FC of the $k$-th relay as
\begin{eqnarray}
{\bf{H}}_k&=&\widehat{\bf{H}}_k+e_1{\bf{\Omega}}_{1,k},\label{revision1}\\
{\bf{G}}_k&=&\widehat{\bf{G}}_k+e_2{\bf{\Omega}}_{2,k},\label{revision2}
\end{eqnarray}
where ${\bf{H}}_k  \in\mathbb{C}^{M \times M}$ and ${\bf{G}}_k
\in\mathbb{C}^{M \times M}$ $(k=1,...,K)$ stand  for the backward
and forward MIMO channel matrix of the $k$-th relay node
respectively. ${\bf{\Omega}}_{1,k}$ and ${\bf{\Omega}}_{2,k}$ are
matrices respectively independent of ${\bf{H}}_k$ and ${\bf{G}}_k$,
whose entries are $i.i.d$ zero-mean complex Gaussian, with unity
variance~\cite{14,27}. Therefore the power of CSI errors of BC and FC
are $e_1^{2}$ and $e_2^{2}$. Since $e_1^2$ is the power of channel estimation error, it can be
made very small. $e_2^2$ is the power of the channel error majorally coming from channel quantization,
which is bounded by $2^{-B/M}$ if $B$ bits is used to do quantization. In this paper, we therefore assume $e_1^2\ll 1$ and $e_2^2<1$,
which are reasonable assumptions in a practical system. In this paper, all the relay nodes are supposed to
be located in a cluster. Then all the channels
${\bf{H}}_1,\cdots,{\bf{H}}_K$ and ${\bf{G}}_1,\cdots,{\bf{G}}_K$
can be supposed to be independently and identically distributed
($i.i.d$) and experience the same Rayleigh flat fading. Assume that
the entries of ${\bf{H}}_k$ and ${\bf{G}}_k$ are zero-mean complex
Gaussian random variables with variance one.
\begin{figure}
\centering
\includegraphics[width=3in]{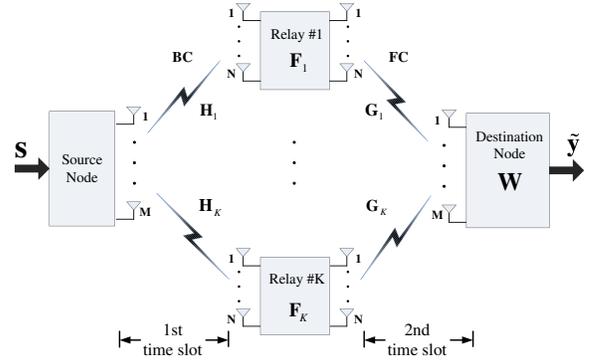}
\caption{System model of the dual-hop MIMO multi-relay network with relay beamforming and successive interference cancellation (SIC) at the destination.}
\end{figure}
In the first time slot, the source node broadcasts the signal to all
the relay nodes through BCs. Let $M{\times}1$ vector $\mbox{\boldmath
$\mathbf{s}$}$ be the transmit signal vector satisfying the power
constraint ${\rm E}\left\{ {{\bf{ss}}^H } \right\} = \left( {{P
\mathord{\left/
 {\vphantom {P M}} \right.
 \kern-\nulldelimiterspace} M}} \right){\bf{I}}_M$, where $P$ is defined
as the transmit power at the source node. Then the corresponding received signal at the $k$-th
relay can be written as
\begin{equation}
{\bf{r}}_k  = {\bf{H}}_k {\bf{s}} + {\bf{n}}_k,
\end{equation}
where the term $\mbox{\boldmath$\mathbf{n}$}_{k}$ is the
spatio-temporally white zero-mean complex additive Gaussian noise
vector, independent across $k$, with the covariance matrix ${\rm
E}\left\{ {{\bf{n}}_k {\bf{n}}_k^H } \right\} = \sigma _1^2
{\bf{I}}_M$. Therefore, noise variance $\sigma _1^2$ represents the
noise power at each relay node.

In the second time slot, firstly each relay node performs linear
processing by multiplying ${\bf{r}}_k$ with an $N \times N$ beamforming
matrix ${\bf{F}}_k$. This ${\bf{F}}_k$ is based on its imperfect CSIs ${\widehat{\bf{H}}}_k$ and ${\widehat{\bf{G}}}_k$.
Consequently, the signal vector sent from the $k$-th relay node is
\begin{equation}
{\bf{t}}_k  = {\bf{F}}_k {\bf{r}}_k.
\end{equation}
From more practical consideration, we assume that each relay node
has its own power constraint satisfying ${\rm E}\left\{
{{\bf{t}}_k^H {\bf{t}}_k } \right\} \leq Q$, which is independent of
power $P$. Hence a power constraint condition of ${\bf{t}}_k$ can be
derived as
\begin{equation}\label{9}
\rho\left( {{\bf{t}}_k } \right) = tr\left\{ {{\bf{F}}_k \left(
{\frac{P}{M}{\bf{H}}_k {\bf{H}}_k^H  + \sigma _1^2 {\bf{I}}_N }
\right){\bf{F}}_k^H } \right\} \le Q.
\end{equation}
After linear relay beamforming processing, all the relay nodes forward their data simultaneously to the destination.
Thus the signal vector received by the destination can be expressed as
\begin{multline}\label{rzf5}
 {\bf{y}}  = \sum_{k=1}^{K}{\bf{G}}_k{\bf{t}}_k
+ {\bf{n}}_d  \\=  \sum_{k=1}^{K} {\bf{G}}_k {\bf{F}}_k  {\bf{H}}_k
{\bf{s}} +\sum_{k=1}^{K} {\bf{G}}_k {\bf{F}}_k {\bf{n}}_k
+{\bf{n}}_d,
\end{multline}
where  ${\bf{n}}_d \in \mathbb{C}^M$, satisfying ${\rm
E}\left\{ {{\bf{n}}_d {\bf{n}}_d^H } \right\} = \sigma _2^2
{\bf{I}}_M$, denotes the zero-mean white circularly symmetric
complex additive Gaussian noise vector at the destination node with the
noise power $\sigma _2^2$.

\section{Relay Beamforming Design}
In this section, the QR detector at the destination node for SIC detection is introduced and a relay beamforming scheme based on MMSE receiver and RZF precoder is proposed.
\subsection{QR Decomposition and SIC Detection}
QR-decomposition (QRD) detector is utilized
as the destination receiver $\bf{W}$ in this paper, which is proved
to be asymptotically equivalent to that of the maximum-likelihood detector (MLD)~\cite{13}.
Let $\sum_{k=1}^{K}{\widehat{\bf{G}}}_k {\bf{F}}_k {\widehat{\bf{H}}}_k={\bf{H}}_{\mathcal {S}\mathcal {D}} $ be the effective channel between the source
and destination node, which can be estimated at the destination node by using the AF relay channel estimation methods~\cite{wen1, wen2, wen3}. Then (\ref{rzf5}) can be rewritten as
\begin{equation}
 {\bf{y}}= {\bf{H}}_{\mathcal {S}\mathcal {D}} {\bf{s}}
+ \widehat{\bf{n}},
\end{equation}
where
\begin{multline}\label{11}
\widehat{\bf{n}}\cong\sum_{k=1}^{K} e_1\widehat{\bf{G}}_k {\bf{F}}_k  {\bf{\Omega}}_{1,k}
{\bf{s}} +\sum_{k=1}^{K} e_2{\bf{\Omega}}_{2,k} {\bf{F}}_k  \widehat{\bf{H}}_k
{\bf{s}}\\ +\sum_{k=1}^K {\bf{G}}_k {\bf{F}}_k {\bf{n}}_k
+{\bf{n}}_d
\end{multline}
is the effective noise vector cumulated from the CSI errors, the
noise $\bf{n}_k$ at the $k$-th relay node, and the noise vector
${\bf{n}}_d$ at the destination. In the derivation, we omit the
term including $e_1e_2$. Even if we retain the term $e_1e_2$ in (\ref{11}), it will result in some terms involving
$e_1^2e_2$, $e_1e_2^2$ and $e_1^2e_2^2$ when calculating the covariance of the effective noise $\widehat{\bf{n}}$. The first two terms are always zero after taking expectation,
while the only terms left are those involving $e_1^2e_2^2$. Since $e_1^2\ll 1$ and $e_2^2<1$, we have $e_1^2e_2^2\ll 1$. Therefore, it is reasonable to omit the term including $e_1e_2$ in (\ref{11}).

Finally, in order to cancel the
interference from other antennas, QR decomposition of the effective
channel is implemented as
\begin{equation}\label{5}
  {\bf{H}}_{\mathcal {S}\mathcal {D}}={\bf{Q}}_{\mathcal {S}\mathcal {D}} {\bf{R}}_{\mathcal {S}\mathcal {D}},
\end{equation}
where ${\bf{Q}}_{\mathcal {S}\mathcal {D}}$ is an $M \times M$
unitary matrix and ${\bf{R}}_{\mathcal {S}\mathcal {D}}$ is an $M
\times M$ right upper triangular matrix. Therefore the QRD  detector
at destination node is chosen as: ${\bf{W}}={\bf{Q}}_{\mathcal
{S}\mathcal {D}}^H $, and the signal vector after QRD detection becomes
\begin{equation}
 \tilde {\bf{y}}={\bf{Q}}_{\mathcal {S}\mathcal {D}}^H {\bf{y}}= {\bf{R}}_{\mathcal {S}\mathcal {D}} {\bf{s}}
+ {\bf{Q}}_{\mathcal {S}\mathcal {D}}^H \widehat{{\bf{n}}}.
\end{equation}
%The authors in~\cite{10} calculate the SNR on each link after QR
%decomposition.  Because matrix ${\bf{Q}}_{\mathcal {S}\mathcal
%{D}}^H$ is an unitary matrix, it does not change the power of the
%effective signal and noise, no matter with channel estimation error
%or not. Thus, the SNR will not vary after QR detection. In the
%following part, we can only focus on the SNR before QR
%decomposition.

A power control factor ${\rho_k} $ is set with $\bf{F}_k$ in (\ref{9}) to
guarantee that the $k$-th relay transmit power is equal to $Q$. The transmit signal from each
relay node after linear beamforming and power control becomes
\begin{equation}
 {\bf{t}}_k  = {\rho_k} {\bf{F}}_k {\bf{r}}_k,
\end{equation}
where the power control factor ${\rho_k}$ can be derived from (\ref{9}) as
\begin{equation}\label{30}
{\rho _k}  = \left ( \frac{Q}{\mathrm{E}\left [\mathrm{tr} \left \{ {\bf{F}}_k \left( {\frac{P}{M}{\bf{H}}_k {\bf{H}}_k^H  + \sigma _1^2 {\bf{I}}_N } \right){\bf{F}}_k^H \right \}\right ]} \right ) ^{\frac{1}{2}}.
\end{equation}
\subsection{Beamforming at Relay Nodes}
The MF beamformer is used in~\cite{10} according to maximum ratio transmission (MRT) and maximum ratio combining (MRC) which are advantageous to the beamformers based on matrix decomposition in~\cite{9}. If MF beamformer is used, then
\begin{equation}
{\bf{F}}_k^{\mathrm{MF-MF}}  =  {\widehat{\bf{G}}}_k^H {\widehat{\bf{H}}}_k^H.
\end{equation}
Another choice is to diagonalize the effective channel between the source and destination, for example,
\begin{multline}
{\bf{F}}_k^{\mathrm{ZF-ZF}}  =  {\widehat{\bf{G}}}_k^{\dag}{\widehat{\bf{H}}}_k^{\dag}\\={\widehat{\bf{G}}}_k^H \left( {\widehat{\bf{G}}_k {\widehat{\bf{G}}}_k^H } \right)^{ - 1} \left({\widehat{\bf{H}}}_k^H {\widehat{\bf{H}}}_k\right)^{-1}{\widehat{\bf{H}}}_k^H.
\end{multline}
MF-MF outperforms ZF-ZF in low SNR, while ZF-ZF outperforms MF-MF in high SNR~\cite{10}. But these two schemes are not optimized.

In this paper, we propose a robust MMSE-RZF beamformer at the relay nodes.
When MMSE-RZF is chosen, beamforming at the $k$-th relay is
\begin{multline}\label{28}
{\bf{F}}_k^{\mathrm{MMSE - RZF}}  =  {\widehat{\bf{G}}}_k^H \left( {\widehat{\bf{G}}}_k {\widehat{\bf{G}}}_k^H  + \alpha_k^{\mathrm{RZF}} {\bf{I}}_M  \right)^{ - 1}\\ \left({\widehat{\bf{H}}}_k^H {\widehat{\bf{H}}}_k+ \alpha_k^{\mathrm{MMSE}} {\bf{I}}_M \right)^{-1}{\widehat{\bf{H}}}_k^H.
\end{multline}
Note that MF-MF and ZF-ZF are two extreme cases for  $\alpha_k^{\mathrm{MMSE}}=\alpha_k^{\mathrm{RZF}}=\infty$ and $\alpha_k^{\mathrm{MMSE}}=\alpha_k^{\mathrm{RZF}}=0$ respectively.  Generally, if the $\alpha$ (either regularizing factor in MMSE or RZF) is too large, the effective channel matrix will far deviate from a diagonal matrix, which results in power consumption and interference across different datas. If $\alpha$ is too small, the MMSE receiver and RZF precoder will perform like a ZF receiver or precoder which have the power penalty problem due to its inverse Wishart distribution term in its transmit power~\cite{22,RZF,28}. We aim to obtain the optimal $\alpha_k^{\mathrm{MMSE}}$ and $\alpha_k^{\mathrm{RZF}}$ to maximize the rate in this paper. However, to directly get the global optimal closed-form solution is difficult. In the following, we derive an optimized  solution by two steps. We first derive an optimized $\alpha_k^{\mathrm{MMSE}}$ by maximizing the SINR at the relay nodes, and then we derive an optimized $\alpha_k^{\mathrm{RZF}}$  dependent on the given optimized  $\alpha_k^{\mathrm{MMSE}}$ by maximizing the rate at the destination.

\section{Robust MMSE-RZF Beamformer}
In this section, we derive the optimized $\alpha_k^{\mathrm{MMSE}}$ and $\alpha_k^{\mathrm{RZF}}$ in the MMSE-RZF beamformer by two steps. We first derive the optimized $\alpha_k^{\mathrm{MMSE}}$ by maximizing the SINR at relay nodes, and then derive the optimized $\alpha_k^{\mathrm{RZF}}$ for a given $\alpha_k^{\mathrm{MMSE}}$ based on the asymptotic rate. Although the derived solution is not global optimum, it is observed quite efficient in terms of rate in the simulations.

\subsection{Optimization of $\alpha_k^{\mathrm{MMSE}}$}
We optimize $\alpha_k^{\mathrm{MMSE}}$ by maximizing the SINR at relay nodes.
For the $k$-th relay, the signal vector after processed by an MMSE receiver is
\begin{equation}\label{10}
\begin{split}
{\bf{v}}_k  &= \left({\widehat{\bf{H}}}_k^H {\widehat{\bf{H}}}_k+ \alpha_k^{\mathrm{MMSE}} {\bf{I}}_M \right)^{-1}{\widehat{\bf{H}}}_k^H {\bf{r}}_k\\&=\left({\widehat{\bf{H}}}_k^H {\widehat{\bf{H}}}_k+ \alpha_k^{\mathrm{MMSE}} {\bf{I}}_M \right)^{-1}{\widehat{\bf{H}}}_k^H {\widehat{\bf{H}}}_k {\bf{s}}\\
&+ e_1\left({\widehat{\bf{H}}}_k^H {\widehat{\bf{H}}}_k+ \alpha_k^{\mathrm{MMSE}} {\bf{I}}_M \right)^{-1}{\widehat{\bf{H}}}_k^H {\bf{\Omega}}_{1,k}
{\bf{s}}\\&+\left({\widehat{\bf{H}}}_k^H {\widehat{\bf{H}}}_k+ \alpha_k^{\mathrm{MMSE}} {\bf{I}}_M \right)^{-1}{\widehat{\bf{H}}}_k^H{\bf{n}}_k.
\end{split}
\end{equation}

The first term in (\ref{10}) is the signal vector, which contains inter-stream interference, because matrix $\left({\widehat{\bf{H}}}_k^H {\widehat{\bf{H}}}_k+ \alpha_k^{\mathrm{MMSE}} {\bf{I}}_M \right)^{-1}{\widehat{\bf{H}}}_k^H {\widehat{\bf{H}}}_k$ is not diagonal if $\alpha_k^{\mathrm{MMSE}}\neq 0$. So we need to calculate the power of desired signal and the interference.
We use the diagonal decompositions in the following analysis, i.e.,
\begin{eqnarray}
\widehat{\mathbf{H}}_k{\widehat{\mathbf{H}}}_k^H
&=&\mathbf{P}_k\mathrm{diag}\{\theta_{k,1},\ldots,\theta_{k,M}\}{\mathbf{P}_k}^H
\triangleq \mathbf{P}_k\mathbf{\Theta}_k{\mathbf{P}_k}^H,\nonumber\\
\widehat{\mathbf{G}}_k{\widehat{\mathbf{G}}}_k^H
&=&\mathbf{Q}_k\mathrm{diag}\{\lambda_{k,1},\ldots,\lambda_{k,M}\}{\mathbf{Q}_k}^H
\triangleq \mathbf{Q}_k\mathbf{\Lambda}_k{\mathbf{Q}_k}^H,\nonumber
\end{eqnarray}
where $\mathbf{P}_k$ and $\mathbf{Q}_k$ are unitary matrices.
%Then the power control factor at the $k$-th relay is
%\begin{equation}
%{\rho _k}  =\left ( \frac{Q}{ \mathrm{E}\left \{\mathrm{tr} \left \{
%{\bf{v}}_k{\bf{v}}_k^H \right \}\right \}} \right ) ^{\frac{1}{2}}
%=\left ( \frac{Q}{
%\frac{P}{M}\sum\frac{\theta_{k,n}^2}{(\theta_{k,n}+\alpha^{\mathrm{MMSE}})^2}+\left(e_1^2P+\sigma_1^2\right)\sum\frac{\theta_{k,n}}{(\theta_{k,n}+\alpha^{\mathrm{MMSE}})^2}}
%\right ) ^{\frac{1}{2}}.
%\end{equation}
%Since the expectation is taken to the unnormalized transmit power,
%all the power control factors at each relay are the same.
To divides
the interference from the desired signal, we introduce the following
two lemmas.
\newtheorem{theorem}{Theorem}
\newtheorem{lemma}[theorem]{Lemma}
\begin{lemma}
Assume that $\mathbf{A}\in \mathbb{C}^{M\times M}$ is a random matrix. If there is a  diagonal decomposition $\mathbf{A}=\mathbf{Q}\mathbf{\Lambda}{\mathbf{Q}}^H$, where  $\mathbf{\Lambda}=\mathrm{diag}\{\lambda_1,\ldots,\lambda_M\}\in\mathbb{R}^{M\times M}$ and the matrix $\mathbf{Q}$ is unitary, we have
\begin{multline}
\mathrm{E}\{\left(\mathbf{A}\right)_{m,m}^2\}\\=\frac{1}{M(M+1)}\left(\left(\sum_{\ell=1}^M\lambda_\ell\right)^2+\sum_{\ell=1}^M \lambda_\ell^2\right)\triangleq\mu(\lambda),
\end{multline}
for any $m$, where the conditional expectation is taken with respect to the distribution $\mathbf{Q}$ conditioned on $\mathbf{\Lambda}$.
\end{lemma}

The proof of Lemma 1 can be directly obtained from~\cite{RZF} which considers perfect CSI. Although matrix $\bf{A}$ in this paper is a multiplication of an imperfect channel matrix and its conjugate transpose, whose entries have covariance  $1-e_1^2$ or $1-e_2^2$, the distribution of $\bf{Q}$ is not changed. So is the expectation in Lemma 1. Note that the conditional expectation is taken with respect to $\mathbf{Q}$ conditioned on $\mathbf{\Lambda}$ is valid because $\mathbf{Q}$ and $\mathbf{\Lambda}$ are independent~\cite{independent}.
\begin{lemma}
Assume that $\mathbf{A}\in \mathbb{C}^{M\times M}$ is a random matrix. If there is a  diagonal decomposition $\mathbf{A}=\mathbf{Q}\mathbf{\Lambda}{\mathbf{Q}}^H$, with $\mathbf{\Lambda}=\mathrm{diag}\{\lambda_1,\ldots,\lambda_M\}\in\mathbb{R}^{M\times M}$ and unitary matrix $\mathbf{Q}$, we have
\begin{multline}
\mathrm{E}\{|\left(\mathbf{A}\right)_{m,j}|^2\}=\frac{1}{(M-1)(M+1)}\sum_{\ell=1}^M\lambda_\ell^2\\-\frac{1}{(M-1)M(M+1)}\left(\sum_{\ell=1}^M \lambda_\ell\right)^2\triangleq\nu(\lambda),
\end{multline}
for any $m\neq j$,
where the conditional expectation is taken with respect to the distribution $\mathbf{Q}$ conditioned on $\mathbf{\Lambda}$.
\end{lemma}
\begin{proof}
Because $\mathbf{A}$ is a conjugate symmetric matrix, the conditional expectation with respect to the distribution $\mathbf{Q}$ is
\begin{multline}
\mathrm{E}\left\{\sum_{j=1,j\neq m}^M |\left(\mathbf{A}\right)_{m,j}|^2\right\}+\mathrm{E}\left\{\left(\mathbf{A}\right)_{m,m}^2\right\}\\
=\mathrm{E}\left\{\left(\mathbf{A}\mathbf{A}^H\right)_{m,m}\right\}=\mathrm{E}\left\{\left(\mathbf{Q}\mathbf{\Lambda}^2\mathbf{Q}^H\right)_{m,m}\right\}\\=\frac{1}{M}\sum_{\ell=1}^M\lambda_\ell^2.
\end{multline}
Since $\mathrm{E}\left\{|\left(\mathbf{A}\right)_{k,j}|^2\right\}$ are all equal for $j\neq k$, we have
\begin{eqnarray}
\mathrm{E}\left\{|\left(\mathbf{A}\right)_{k,j}|^2\right\}&=&\frac{1}{(M-1)}\left(\frac{1}{M}\sum_{\ell=1}^M\lambda_\ell^2-\mathrm{E}\{\left(\mathbf{A}\right)_{m,m}^2\}\right)
\nonumber\\&=&\frac{1}{(M-1)(M+1)}\sum_{l=1}^M\lambda_\ell^2\nonumber\\
&&-\frac{1}{(M-1)M(M+1)}\left(\sum_{\ell=1}^M \lambda_\ell\right)^2.
\end{eqnarray}
\end{proof}

Now we return to derive the signal-to-interference noise ratio (SINR) for each stream at each relay.
The first term in the right hand side of (\ref{10}) can be rewritten as
\begin{multline}\label{23}
\left({\widehat{\bf{H}}}_k^H {\widehat{\bf{H}}}_k+ \alpha_k^{\mathrm{MMSE}} {\bf{I}}_M \right)^{-1}{\widehat{\bf{H}}}_k^H {\widehat{\bf{H}}}_k {\bf{s}}\\=\mathbf{P}_k\frac{\mathbf{\Theta}_k}{\mathbf{\Theta}_k+\alpha_k^{\mathrm{MMSE}}\mathbf{I}_N}{\mathbf{P}_k}^H \mathbf{s}.
\end{multline}
Therefore, from Lemma 1, the power of the desired signal of the $m$-th stream can be calculated by conditional expectation as
\begin{multline}\label{20}
\mathrm{E}\left\{\left|\left(\mathbf{P}_k\frac{\mathbf{\Theta}_k}{\mathbf{\Theta}_k+\alpha_k^{\mathrm{MMSE}}\mathbf{I}_N}{\mathbf{P}_k}^H \right)_{m,m}\mathbf{s}_m\right|^2\right\}\\=\frac{P}{M}\mu\left(\frac{\theta_k}{\theta_k+\alpha_k^{\mathrm{MMSE}}}\right),
\end{multline}
where $\theta_k$ denotes the set of all the diagonal entries in $\mathbf{\Theta}_k$.
From Lemma 2, the interference from other streams by conditional expectation are
\begin{multline}\label{21}
\mathrm{E}\left\{\left|\sum_{j=1,j\neq m}^M\left(\mathbf{P}_k\frac{\mathbf{\Theta}_k}{\mathbf{\Theta}_k+\alpha_k^{\mathrm{MMSE}}\mathbf{I}_N}{\mathbf{P}_k}^H \right)_{m,j}\mathbf{s}_j\right|^2\right\}\\=\frac{P(M-1)}{M}\nu\left(\frac{\theta_k}{\theta_k+\alpha_k^{\mathrm{MMSE}}}\right).
\end{multline}
The effective noise of the $m$-th stream is
\begin{multline}
\mathbf{n}_{\mathrm{eff},k}=e_1\left({\widehat{\bf{H}}}_k^H {\widehat{\bf{H}}}_k+ \alpha_k^{\mathrm{MMSE}} {\bf{I}}_M \right)^{-1}{\widehat{\bf{H}}}_k^H {\bf{\Omega}}_{1,k}
{\bf{s}}\\+\left({\widehat{\bf{H}}}_k^H {\widehat{\bf{H}}}_k+ \alpha_k^{\mathrm{MMSE}} {\bf{I}}_M \right)^{-1}{\widehat{\bf{H}}}_k^H{\bf{n}}_k,
\end{multline}
whose covariance matrix by conditional expectation can be calculated as
\begin{multline}\label{22}
\mathrm{E}\left\{\mathbf{n}_{\mathrm{eff},k}\mathbf{n}_{\mathrm{eff},k}^H\right\}
=(e_1^2P+\sigma_1^2)\\\times\mathrm{E}\left\{\mathrm{diag}\left\{\left\{\left(\mathbf{P}_k\frac{\mathbf{\Theta}_k}{\left(\mathbf{\Theta}_k
+\alpha_k^{\mathrm{MMSE}}\right)^2}\mathbf{P}_k^H\right)_{\ell,\ell}\right\}_{\ell=1}^M\right\}\right\}\\
=\frac{e_1^2P+\sigma_1^2}{M}\sum_{\ell=1}^M\frac{\theta_{k,\ell}}{\left(\theta_{k,\ell}+\alpha_k^{\mathrm{MMSE}}\right)^2}\mathbf{I}_M,
\end{multline}
where we used the fact $\mathrm{E}\left \{\mathbf{\Omega}\mathbf{A}\mathbf{\Omega}^H\right \}=\mathrm{tr} \left (\mathbf{A}\right )$ for an $N\times N$ matrix $\mathbf{A}$ and a unitary random matrix $\mathbf{\Omega}$. In (\ref{20}), (\ref{21}) and (\ref{22}), the conditional expectations are taken with respect to their respective unitary matrices.
Combining (\ref{20}), (\ref{21}), and (\ref{22}), the SINR of the $m$-th stream at the $k$-th relay is
\begin{multline}\label{SINR}
\mathrm{SINR}_{k,m}^{\mathrm{R}}\\
=\frac{\frac{P}{M}\mu\left(\frac{\theta_k}{\theta_k+\alpha_k^{\mathrm{MMSE}}}\right)}
{\frac{P(M-1)\nu\left(\frac{\theta_k}{\theta_k+\alpha_k^{\mathrm{MMSE}}}\right)}{M}+\frac{\sum_{\ell=1}^M\frac{(e_1^2P+\sigma_1^2)\theta_{k,\ell}}{\left(\theta_{k,\ell}+\alpha_k^{\mathrm{MMSE}}\right)^2}}{M}}.
\end{multline}

%At the destination, the received vector is from all the $K$ relays. Since the expectation of the non-diagonal entries in the effective channel matrix for the $k$-th relay in (\ref{23}) is zero, so as $K$ increases, the interference are approximately zero, i.e.,
%\begin{equation}
%\left(\sum_{k=1}^K\mathbf{P}_k\frac{\mathbf{\Theta}_k}{\mathbf{\Theta}_k+\alpha_k^{\mathrm{MMSE}}\mathbf{I}_N}{\mathbf{P}_k}^H\right)_{j,m (j\neq m)}\approx 0.
%\end{equation}
%Therefore, after QR decomposition, the interference is also approximately zero since the unitary matrix $\mathbf{Q}$ is almost an identity matrix in this case.
%The desired signal and noise inherited from the relays are scaled by $K^2$ and $K$ respectively.
%Therefore, we have the SNR of the $m$-th stream by idealizing the forward channels as
%\begin{equation}\label{19}
%\mathrm{SNR}_{m}^{\mathrm{D}}\approx \frac{\frac{PK^2}{M}\mu\left(\frac{\theta}{\theta+\alpha^{\mathrm{MMSE}}}\right)}
%{K\frac{e_1^2P+\sigma_1^2}{M}\sum\frac{\theta}{\left(\theta+\alpha^{\mathrm{MMSE}}\right)^2}+\sigma_2^2\rho^{-2}},
%\end{equation}
%where the power control factor $\rho$ at relay normalizes the noise at the destination.
The derived SINR in (\ref{SINR}) is neither the instantaneous SINR, nor the average SINR. It is the average SINR over the channels corresponding to the fixed Eigenmode $\mathbf{\Theta}$. To maximize the SINR expression, we introduce the following lemma which is a conclusion of the Appendix B in~\cite{RZF}.
\begin{lemma}
For an SNR in terms of $\alpha$,
\begin{multline}\label{15}
\mathrm{SNR}(\alpha)\\=\frac{A\left(\sum_{\ell=1}^M\frac{\lambda_\ell}{\lambda_\ell+\alpha}\right)^2+B\sum_{\ell=1}^M\frac{\lambda_\ell^2}{(\lambda_\ell+\alpha)^2}}{\sum_{\ell=1}^M\left[\frac{C\lambda_l}{(\lambda_\ell+\alpha)^2}+\frac{D\lambda_\ell^2}{(\lambda_\ell+\alpha)^2}+E\left(\frac{\lambda_\ell}{\lambda_\ell+\alpha}\right)^2\right]},
\end{multline}
is maximized by $\alpha=C/D$.
\end{lemma}

The optimum value of $\alpha$ can be obtained by differentiating
(\ref{15}) and setting it to be zero, which  results in
\begin{equation}
\sum_{\ell>k}\frac{\lambda_\ell\lambda_k(\lambda_k-\lambda_\ell)^2(C/D-\alpha)}{(\lambda_\ell+\alpha)^3(\lambda_k+\alpha)^3}=0.
\end{equation}
Since the eigenvalues are not all equal, the SINR is maximized only when $\alpha=C/D$.

Substituting $\mu(\lambda)$ and $\nu(\lambda)$ into (\ref{SINR}) and using Lemma 3, we obtain
\begin{multline}\label{29}
\alpha_k^{\mathrm{MMSE,opt}}=\frac{\frac{e_1^2P+\sigma_1^2}{M}}{\frac{1}{(M-1)(M+1)}\cdot\frac{P(M-1)}{M}}\\ =(M+1)\left(e_1^2+\frac{\sigma_1^2}{P}\right).
\end{multline}

We see that the derived $\alpha_k^{\mathrm{MMSE,opt}}$ is a closed-form value independent of the instantaneous channel. It is a function of the power of CSI error ($e_1^2$) and the SNR ($P/\sigma_1^2$) of the BC. $\alpha_k^{\mathrm{MMSE,opt}}$ increases with $e_1^2$, which means that a large regularization is needed to balance the desired signal and the additional noise inherited from the CSI error.
\subsection{Optimization of $\alpha_k^{\mathrm{RZF}}$}
To optimize $\alpha_k^{\mathrm{RZF}}$, we need to derive the rate of the system. In the rest of the analysis, we write ${\bf{F}}_k^{\mathrm{MMSE-RZF}}$ in (\ref{28}) as ${\bf{F}}_k$ for simplicity. By adding the power control factor at the relays, we have
\begin{equation}\label{revision3}
{\bf{H}}_{\mathcal {S}\mathcal {D}}=\sum_{k=1}^{K}\rho_k{\widehat{\bf{G}}}_k {\bf{F}}_k {\widehat{\bf{H}}}_k.
\end{equation}
The effective noise vector in (\ref{11}) is
\begin{multline}
\widehat{\bf{n}}=\sum_{k=1}^{K} e_1\rho_k\widehat{\bf{G}}_k {\bf{F}}_k {\bf{\Omega}}_{1,k}
{\bf{s}} \\+\sum_{k=1}^{K} e_2\rho_k{\bf{\Omega}}_{2,k} {\bf{F}}_k  \widehat{\bf{H}}_k
{\bf{s}} +\sum_{k=1}^K \rho_k{\bf{G}}_k {\bf{F}}_k {\bf{n}}_k
+{\bf{n}}_d,\label{12}
\end{multline}
which, after the QR decomposition of the effective channel, has a covariance matrix as
\begin{multline}\label{19}
\mathrm{E}\left\{\widehat{\bf{n}}\widehat{\bf{n}}^H\right\}=
\mathrm{diag}\left\{\left\{\left(e_1^2P+\sigma_1^2\right)\right.\right.\\\left.\left.\sum_{k=1}^K
\left\|\rho_k\left(\mathbf{Q}_{\mathcal{SD}}^H\widehat{{\bf{G}}}_k {\bf{F}}_k\right)_m\right\|^2\right\}_{m=1}^M\right\}\\
+\left(\frac{Pe_2^2}{M}\sum_{k=1}^K\rho_k^2\mathrm{tr}\left(\mathbf{F}_k\widehat{\mathbf{H}}_k\widehat{\mathbf{H}}_k^H\mathbf{F}_k^H\right)\right.\\
\left.+e_2^2\sigma_1^2\sum_{k=1}^K\rho_k^2\mathrm{tr}\left(\mathbf{F}_k\mathbf{F}_k^H\right)+\sigma_2^2\right)\mathbf{I}_M
\triangleq \mathbf{N}_\mathrm{cov}.
\end{multline}

Finally we obtain the SNR of the $m$-th data stream at the destination after QR decomposition as
\begin{equation}\label{13}
\mathrm{SNR}_m^{\mathrm{D}}=\frac{\frac{P}{M}\left|(\mathbf{R}_\mathcal{SD})_{m,m}\right|^2}
{\frac{P}{M}\sum_{j=m+1}^M\left|(\mathbf{R}_\mathcal{SD})_{m,j}\right|^2+\left(\mathbf{N}_\mathrm{cov}\right)_{m,m}}.
\end{equation}
The ergodic rate is derived by summing up all the data rates on each antenna link, i.e.,
\begin{equation}\label{16}
C = \mathrm{E}_{\left\{ {\widehat{\bf{H}}}_k ,{\widehat{\bf{G}}}_k  \right\}_{k = 1}^K}  \left\{ {\frac{1}{2}\sum\limits_{m = 1}^M {\log _2 \left( {1 + \mathrm{SNR}_m^{\mathrm{D}} } \right)} } \right\},
\end{equation}
where the $\frac{1}{2}$ penalty is due to the two time-slot transmission.
From (\ref{13}), we see that it is difficult to obtain the optimal solution directly. We derive
asymptotic rate for large $K$ and then get the optimized $\alpha_k^{\mathrm{RZF}}$. Since all terms in (\ref{revision3}) and (\ref{12}) include $\rho_k$ except for $\bf{n}_d$, we first consider the expectation of $\rho^{-2}_k$.
%
%We take expectations to $\rho_k^{-2}$ which results in a uniform fixed $\rho^{-2}$ for all relays. This manipulation guarantees that the expectation of the power of relays is $Q$. The performance with this assumption varies little compared with using the dynamic power control factors~\cite{swindlehurst}.
%
%From (\ref{revision3}) and (\ref{12}), we observe that all terms includes $\rho_k$ except for $\bf{n}_d$. Therefore in the following analysis, we multiply $\rho^{-2}$ to $\sigma_2^2$ after calculating the power of $\bf{n}_d$ in  (\ref{revision4}) and write other terms without $\rho^{-2}$ for simplicity. It is easy to see that the SINR are the same.
%
From (\ref{30}), substituting the perfect CSIs with (\ref{revision1}) and (\ref{revision2}), and taking the conditional expectation with respect to $\mathbf{P}_k$ and $\mathbf{Q}_k$ and conditioned on $\lambda$ and $\theta$, we have (\ref{53}).
\begin{figure*}
\begin{multline}\label{53}
\mathrm{E}\left\{\rho^{-2}_k\right\}
=\frac{1}{Q}\mathrm{E}\left\{\frac{P}{M}\mathrm{tr}\left(\mathbf{F}_k(\widehat{\mathbf{H}}_k\widehat{\mathbf{H}}_k^H+e_1^2\mathbf{\Omega}_{1,k}\mathbf{\Omega}_{1,k}^H)\mathbf{F}_k^H\right)+\sigma_1^2\mathrm{tr}\left(\mathbf{F}_k\mathbf{F}_k^H\right)\right\}
\\=\frac{P}{QM}\mathrm{E}\left\{\mathrm{tr}\left(\mathbf{Q}_k\frac{\mathbf{\Lambda}_k}{\left(\mathbf{\Lambda}_k+\alpha^{\mathrm{RZF}}\mathbf{I}_M\right)^2}\mathbf{Q}_k^H\mathbf{P}_k\frac{\mathbf{\Theta}_k^2}{\left(\mathbf{\Theta}_k+\alpha^{\mathrm{MMSE}}\mathbf{I}_M\right)^2}\mathbf{P}_k^H\right)\right\}\\+\frac{Pe_1^2+\sigma_1^2}{Q}
\mathrm{E}\left\{\mathrm{tr}\left(\mathbf{Q}_k\frac{\mathbf{\Lambda}_k}{\left(\mathbf{\Lambda}_k+\alpha^{\mathrm{RZF}}\mathbf{I}_M\right)^2}\mathbf{Q}_k^H\mathbf{P}_k\frac{\mathbf{\Theta}_k}{\left(\mathbf{\Theta}_k+\alpha^{\mathrm{MMSE}}\mathbf{I}_M\right)^2}\mathbf{P}_k^H\right)\right\}
\\=\frac{P}{Q}\mathrm{E}\left\{\frac{\theta^2}{(\theta+\alpha^{\mathrm{MMSE}})^2}\right\}\mathrm{E}\left\{\frac{\lambda}{(\lambda+\alpha^{\mathrm{MMSE}})^2}\right\}
+\frac{(e_1^2P+\sigma_1^2)M}{Q}
\\\times\mathrm{E}\left\{\frac{\theta}{(\theta+\alpha^{\mathrm{MMSE}})^2}\right\}\mathrm{E}\left\{\frac{\lambda}{(\lambda+\alpha^{\mathrm{MMSE}})^2}\right\}
\triangleq \rho^{-2},
\end{multline}
\hrulefill
\begin{multline}\label{14}
\left({\bf{H}}_{{\mathcal{S}\mathcal {D}}}\right)_{i,i}\overset{w.p.}{\longrightarrow}K\left(\mathrm{E}\left\{\left({\widehat{\bf{G}}}_k {\bf{F}}_k{\widehat{\bf{H}}}_k\right)_{i,i}\right\}\right)=K\mathrm{E}\biggl\{\left(\mathbf{Q}_k\frac{\mathbf{\Lambda}_k}{\mathbf{\Lambda}_k+\alpha^{\mathrm{RZF}}\mathbf{I}_M}\mathbf{Q}_k^H\mathbf{P}_k\frac{\mathbf{\Theta}_k}{\mathbf{\Theta}_k+\alpha^{\mathrm{MMSE}}\mathbf{I}_M}\mathbf{P}_k^H\right)_{i,i}\biggl\}
\\
=K\mathrm{E}\left\{\left(\mathbf{Q}_k\frac{\mathbf{\Lambda}_k}{\mathbf{\Lambda}_k+\alpha^{\mathrm{RZF}}\mathbf{I}_M}\mathbf{Q}_k^H\right)_{m,m}\right\}\mathrm{E}\left\{\left(\mathbf{P}_k\frac{\mathbf{\Theta}_k}{\mathbf{\Theta}_k+\alpha^{\mathrm{MMSE}}\mathbf{I}_M}\mathbf{P}_k^H\right)_{n,n}\right\}
\\
=\frac{K}{MN}\mathrm{E}\left\{\sum_{m=1}^M\frac{\theta_{k,m}}{\theta_{k,m}+\alpha^{\mathrm{MMSE}}}\right\}\mathrm{E}\left\{\sum_{m=1}^M\frac{\lambda_{k,m}}{\lambda_{k,m}+\alpha^{\mathrm{RZF}}}\right\}
=K\mathrm{E}\left\{\frac{\theta}{\theta+\alpha^{\mathrm{MMSE}}}\right\}\mathrm{E}\left\{\frac{\lambda}{\lambda+\alpha^{\mathrm{RZF}}}\right\},
\end{multline}
\hrulefill
\end{figure*}
%where the expectation is taken with respect to the distributions of $\mathbf{P}_k$ and $\mathbf{Q}_k$ $(k=1,\ldots,K)$.
Here we denote $\alpha$, $\lambda$ and $\theta$ without subscript $k$ and $m$ for simplicity, because all the channels for different relays are $i.i.d.$, and $\lambda_m$ (and $\theta_m$) for every $m$ are identically distributed. (\ref{53}) implies that the expectation of $\rho_k^{-2}$  results in a uniform fixed $\rho^{-2}$ for all relays. Therefore we approximate $\rho_k^{-2}$ by  $\rho^{-2}$ in the following analysis. The performance with such approximation varies little compared with using the dynamic power control factors~\cite{swindlehurst}. Since all terms in the numerator and the denominator of (\ref{13}) excerpt for the $\bf{n}_d$ will generate $\rho_k^{2}$, in the following analysis, we can omit $\rho_k$ in calculation and multiply $\rho^{-2}$ to $\sigma^2$ after calculating the power of $\bf{n}_d$ in (\ref{revision4}).

For the case of large $K$, using Law of Large Number, we have the approximations (\ref{14}).
%where in step (a) we used the fact that $\mathbf{P}_k$, $\mathbf{Q}_k$, $\mathbf{\Theta}_k$ and $\mathbf{\Lambda}_k$ are independent to each other for any $k$ and the expectation of ${\widehat{\bf{G}}}_k{\widehat{\bf{G}}}_k^H$ and ${\widehat{\bf{H}}}_k^H{\widehat{\bf{H}}}_k$ are all diagonal.
Note that
\begin{equation}
\begin{split}
&\mathrm{E}\left\{\left({\widehat{\bf{G}}}_k {\bf{F}}_k{\widehat{\bf{H}}}_k\right)_{i,j}\right\}\\
=&\mathrm{E}\biggl\{\left(\mathbf{Q}_k\frac{\mathbf{\Lambda}_k}{\mathbf{\Lambda}_k+\alpha^{\mathrm{RZF}}\mathbf{I}_M}\mathbf{Q}_k^H\mathbf{P}_k\frac{\mathbf{\Theta}_k}{\mathbf{\Theta}_k+\alpha^{\mathrm{MMSE}}\mathbf{I}_M}\mathbf{P}_k^H\right)_{i,j}\biggl\}\\
=&\sum_{\ell,m,n}\mathrm{E}\biggl\{\left(\mathbf{Q}_k\right)_{i,\ell}\left(\frac{\mathbf{\Lambda}_k}{\mathbf{\Lambda}_k+\alpha^{\mathrm{RZF}}\mathbf{I}_M}\right)_{\ell}\left(\mathbf{Q}_k^H\right)_{\ell,m}\\
&\times\left(\mathbf{P}_k\right)_{m,n}
\left(\frac{\mathbf{\Theta}_k}{\mathbf{\Theta}_k+\alpha^{\mathrm{MMSE}}\mathbf{I}_M}\right)_{n}\left(\mathbf{P}_k^H\right)_{n,j}\biggl\}\\
=&0
\end{split}
\end{equation}
for $i\neq j$,
because $\left(\mathbf{Q}_k\right)_{i,l}\left(\mathbf{Q}_k^H\right)_{l,m}=0$ for $i\neq m$ and $\left(\mathbf{P}_k\right)_{m,n}\left(\mathbf{P}_k^H\right)_{n,j}=0$ for $m\neq j$. Then we have
\begin{multline}\label{52}
\frac{\left({\bf{H}}_{{\mathcal{S}\mathcal {D}}}\right)_{i,j}}{K}=\frac{\sum_{k=1}^K\left({\widehat{\bf{G}}}_k {\bf{F}}_k{\widehat{\bf{H}}}_k\right)_{i,j}}{K}\\\overset{w.p.}{\longrightarrow}\mathrm{E}\left\{\left({\widehat{\bf{G}}}_k {\bf{F}}_k{\widehat{\bf{H}}}_k\right)_{i,j}\right\}=0
\end{multline}
for large $K$. Therefore, from (\ref{14}) and (\ref{52}), we have
\begin{eqnarray}
\left({\bf{H}}_{{\mathcal{S}\mathcal {D}}}\right)_{i,i}&=&{O}(K)\\
\left({\bf{H}}_{{\mathcal{S}\mathcal {D}}}\right)_{i,j}&=&{o}(K),
\end{eqnarray}
which results in that $\frac{{\bf{H}}_{{\mathcal{S}\mathcal {D}}}}{K}$  is asymptotically diagonal for large $K$. So we have
$\mathbf{Q}_{\mathcal{SD}}\overset{w.p.}{\longrightarrow}\mathbf{I}_M$ and $\mathbf{R}_{\mathcal{SD}}\overset{w.p.}{\longrightarrow}\mathbf{H}_{\mathcal{SD}}$ for large $K$.

%where in step (a) we used the fact that $\mathbf{P}_k$, $\mathbf{Q}_k$, $\mathbf{\Theta}_k$ and $\mathbf{\Lambda}_k$ are independent to each other for any $k$ and the expectation of ${\widehat{\bf{G}}}_k{\widehat{\bf{G}}}_k^H$ and ${\widehat{\bf{H}}}_k^H{\widehat{\bf{H}}}_k$ are all diagonal. The expectation is taken over distributions of $\mathbf{P}_k$ and $\mathbf{Q}_k$ $(k=1,\ldots,K)$.
To obtain the power of interference, we calculate the non-diagonal entries of the effective channel matrix which is included in (\ref{51}) in the appendix.
Let us define the following expectations.
\begin{eqnarray*}
&&{\mathcal{E}}_1^\theta\triangleq\mathrm{E}\left\{\frac{\theta}{(\theta+\alpha^{\mathrm{MMSE}})}\right\},\\
&&{\mathcal{E}}_2^\theta\triangleq\mathrm{E}\left\{\frac{\theta}{(\theta+\alpha^{\mathrm{MMSE}})^2}\right\},\\
&&{\mathcal{E}}_3^\theta\triangleq\mathrm{E}\left\{\frac{\theta^2}{(\theta+\alpha^{\mathrm{MMSE}})^2}\right\},\\
&&{\mathcal{E}}_4^\theta\triangleq\mathrm{E}\left\{\frac{\theta\theta^{\prime}}{(\theta+\alpha^{\mathrm{MMSE}})(\theta^{\prime}+\alpha^{\mathrm{MMSE}})}\right\},
\end{eqnarray*}
and
\begin{eqnarray*}
&&{\mathcal{E}}_1^\lambda\triangleq\mathrm{E}\left\{\frac{\lambda}{(\lambda+\alpha^{\mathrm{RZF}})}\right\},\\
&&{\mathcal{E}}_2^\lambda\triangleq\mathrm{E}\left\{\frac{\lambda}{(\lambda+\alpha^{\mathrm{RZF}})^2}\right\},\\ &&{\mathcal{E}}_3^\lambda\triangleq\mathrm{E}\left\{\frac{\lambda^2}{(\lambda+\alpha^{\mathrm{RZF}})^2}\right\},\\
&&{\mathcal{E}}_4^\lambda\triangleq\mathrm{E}\left\{\frac{\lambda\lambda^{\prime}}{(\lambda+\alpha^{\mathrm{RZF}})(\lambda^{\prime}+\alpha^{\mathrm{RZF}})}\right\}.
\end{eqnarray*}
Substituting (\ref{19}), (\ref{53}), (\ref{14}) and (\ref{51}) into (\ref{13}) and (\ref{16}), we obtain the asymptotic rate of the system as
\begin{equation}\label{1}
%C\overset{w.p.}{\longrightarrow}
%\frac{M}{2}\log_2\left(1+\right.\\\left.\frac{\frac{P}{M}\left(K\mathrm{E}\left\{\frac{\theta}{\theta+\alpha^{\mathrm{MMSE}}}\right\}\mathrm{E}\left\{\frac{\lambda}{\lambda+\alpha^{\mathrm{RZF}}}\right\}\right)^2}
%{\left(e_1^2P+\sigma_1^2\right)K\mathrm{E}\left\{\frac{\theta}{(\theta+\alpha^{\mathrm{MMSE}})^2}\right\}\mathrm{E}\left\{\frac{\lambda^2}{(\lambda+\alpha^{\mathrm{RZF}})^2}\right\}
%+PKe_2^2\mathrm{E}\left\{\frac{\theta^2}{(\theta+\alpha^{\mathrm{MMSE}})^2}\right\}\mathrm{E}\left\{\frac{\lambda}{(\lambda+\alpha^{\mathrm{RZF}})^2}\right\}
%+e_2^2\sigma_1^2KM\mathrm{E}\left\{\frac{\theta}{(\theta+\alpha^{\mathrm{MMSE}})^2}\right\}\mathrm{E}\left\{\frac{\lambda}{(\lambda+\alpha^{\mathrm{RZF}})^2}\right\}
%+\sigma_2^2\rho^{-2}}\right)\\
C\overset{w.p.}{\longrightarrow}\frac{M}{2}\log_2\left(1+\frac{\frac{P}{M}\left(K{\mathcal{E}}_1^\theta\mathcal{E}_1^\lambda\right)^2}
{\mathcal{I}({\mathcal{E}}^\lambda,{\mathcal{E}}^\theta)+\mathcal{N}({\mathcal{E}}^\lambda,{\mathcal{E}}^\theta)}\right)
,
\end{equation}
where
\begin{multline}
\mathcal{I}({\mathcal{E}}^\lambda,{\mathcal{E}}^\theta)\\=
\frac{PK(M-1)(M+2)}{M(M+1)^2}{\mathcal{E}}_3^\theta{\mathcal{E}}_3^\lambda
-\frac{PK(M-1)}{M(M+1)^2}{\mathcal{E}}_4^\theta{\mathcal{E}}_3^\lambda
\\-\frac{PK(M-1)}{M(M+1)^2}{\mathcal{E}}_3^\theta{\mathcal{E}}_4^\lambda
-\frac{PK(M-1)M}{M(M+1)^2}{\mathcal{E}}_4^\theta{\mathcal{E}}_4^\lambda
\end{multline}
is the power of interference, and
\begin{multline}\label{revision4}
\mathcal{N}({\mathcal{E}}^\lambda,{\mathcal{E}}^\theta)
=\left(e_1^2P+\sigma_1^2\right)K{\mathcal{E}}_2^\theta{\mathcal{E}}_3^\lambda
+PKe_2^2{\mathcal{E}}_3^\theta{\mathcal{E}}_2^\lambda\\
+e_2^2\sigma_1^2KM{\mathcal{E}}_2^\theta{\mathcal{E}}_2^\lambda
+\sigma_2^2\rho^{-2}
\end{multline}
is the asymptotic power of noise for large $K$, which can be calculated by taking expectations over $\mathbf{P}_k$ and $\mathbf{Q}_k$ to (\ref{19}).
Generally, the expectations in the asymptotic rate are difficult to obtain. Fortunately, if we write the expectations by the arithmetic mean with random samples $\lambda(\ell)$ for $\ell=0,\dots,\infty$ as
%\begin{equation}
%{\mathcal{E}}_1^\theta=\mathrm{E}\left\{\frac{\theta}{\theta+\alpha^{\mathrm{MMSE}}}\right\}
%=\lim_{L\rightarrow\infty}\frac{1}{L}\Sigma_{l=1}^L\frac{\theta_l}{\theta_l+\alpha^{\mathrm{MMSE}}},
%\end{equation}
%\begin{equation}
%{\mathcal{E}}_2^\theta=\mathrm{E}\left\{\frac{\theta}{(\theta+\alpha^{\mathrm{MMSE}})^2}\right\}
%=\lim_{L\rightarrow\infty}\frac{1}{L}\Sigma_{l=1}^L\frac{\theta_l}{(\theta_l+\alpha^{\mathrm{MMSE}})^2},
%\end{equation}
%\begin{equation}
%{\mathcal{E}}_3^\theta=\mathrm{E}\left\{\frac{\theta^2}{(\theta+\alpha^{\mathrm{MMSE}})^2}\right\}
%=\lim_{L\rightarrow\infty}\frac{1}{L}\Sigma_{l=1}^L\frac{\theta_l^2}{(\theta_l+\alpha^{\mathrm{MMSE}})^2},
%\end{equation}
\begin{eqnarray}
{\mathcal{E}}_1^\lambda
&=&\lim_{L\rightarrow\infty}\frac{1}{L}\Sigma_{\ell=1}^L\frac{\lambda(\ell)}{\lambda(\ell)+\alpha^{\mathrm{RZF}}},\\
{\mathcal{E}}_2^\lambda
&=&\lim_{L\rightarrow\infty}\frac{1}{L}\Sigma_{\ell=1}^L\frac{\lambda(\ell)}{(\lambda(\ell)+\alpha^{\mathrm{RZF}})^2},\\
{\mathcal{E}}_3^\lambda
&=&\lim_{L\rightarrow\infty}\frac{1}{L}\Sigma_{\ell=1}^L\frac{[\lambda(\ell)]^2}{(\lambda(\ell)+\alpha^{\mathrm{RZF}})^2},
\end{eqnarray}
\begin{multline}
{\mathcal{E}}_4^\lambda
=\lim_{L\rightarrow\infty}\frac{1}{L(L-1)}\left(\left(\Sigma_{\ell=1}^L\frac{\lambda(\ell)}{\lambda(\ell)+\alpha^{\mathrm{RZF}}}\right)^2\right.\\
\left.-\Sigma_{\ell=1}^L\frac{[\lambda(\ell)]^2}{(\lambda(\ell)+\alpha^{\mathrm{RZF}})^2}\right),
\end{multline}
the asymptotic rate can be maximized by using Lemma 3. Finally we obtain
\begin{multline}\label{24}
\alpha^{\mathrm{RZF},\mathrm{opt}}=\\\frac{(PKe_2^2+\frac{\sigma_2^2P}{Q})\mathcal{E}_3^\theta+(e_2^2\sigma_1^2KM+\frac{(e_1^2P+\sigma_1^2)M\sigma_2^2}{Q})\mathcal{E}_2^\theta}
{(e_1^2P+\sigma_1^2)K\mathcal{E}_2^\theta+\frac{PK(M-1)(M+2)}{M(M+1)^2}\mathcal{E}_3^\theta-\frac{PK(M-1)}{M(M+1)^2}\mathcal{E}_4^\theta}.
\end{multline}

\subsection{Remarks}
From (\ref{24}) we can observe that $\alpha^{\mathrm{RZF},\mathrm{opt}}$ is independent of the instantaneous CSIs. This is very practical because once we have  $\alpha^{\mathrm{MMSE},\mathrm{opt}}$ in terms of SNRs of BC and FC which we call as PNR ($=P/\sigma_1^2$) and QNR ($=Q/\sigma_2^2$), $e_1^2$ and $e_2^2$, $\alpha^{\mathrm{RZF},\mathrm{opt}}$ can be easily calculated although the solution is not in closed-form. The derived $\alpha^{\mathrm{MMSE},\mathrm{opt}}$ in (\ref{29}) and $\alpha^{\mathrm{RZF},\mathrm{opt}}$ in (\ref{24}) are respectively  monotonically increasing functions of $e_1$ and $e_2$. So the robust MMSE-RZF balances the desired signal and the additional noise inherited from CSI error through larger regularizing factors. On the other hand, the proposed scheme is of low complexity because it needs only one QRD at the destination while in the work of \cite{9} $K$ or $2K$ QRD is needed at the relay nodes.

From the asymptotic rate in (\ref{1}), we see that they satisfies the
scaling law in~\cite{4}, i.e., $C=(M/2) \log(K)+O(1)$ for large $K$. Therefore, the proposed scheme achieves the intranode array gain $M$ and the distributed array gain $K$.
Intranode array gain is the gain obtained from the introduction of
multiple antennas in each node of the dual-hop networks. Distributed
array gain results from the implementation of multiple relay nodes and
needs no cooperation among them.
Note that although the MMSE-RZF with QR SIC is optimized for large $K$, it also has efficient performance for small $K$ which is validated by the simulations.
Also from (\ref{1}), it is observed that when
QNR grows to infinity for a fixed PNR, the rate will
reach a limit. When PNR and QNR both grow to infinity, the capacity will grow linearly with PNR and QNR (dB) for perfect
 CSIs, or reach a limit for imperfect
 CSIs. The limit of the rate performance is always referred as the
``ceiling effect''~\cite{27} and will be confirmed by simulations.

Consider the case where CSI error varies with the
number of relays ($K$). Generally, the CSI errors of BC are caused by the estimation error. The CSI errors of FC are caused by the estimation error, the quantization error and feedback delay. As in~\cite{21},~\cite{27}, and \cite{14}, we assume that
\begin{eqnarray}
e_1^2&=&\sigma_e^2=\frac{1}{1+\frac{\rho_\tau}{M}T_\tau},\label{17}\\
e_2^2&=&\sigma_e^2+\sigma_q^2+\sigma_d^2=\frac{1}{1+\frac{\rho_\tau}{M}T_\tau}+2^{-B/M}\nonumber\\
&&+\left(1-\mathcal{J}_0\left(\frac{K+1}{2}\cdot2\pi f_D \tau\right)\right),\label{18}
\end{eqnarray}
where $\sigma_e^2$
denotes the estimation error during the training phase for TDD mode, $\sigma_q^2$
denotes the quantization error due to limited bits of feedback, and
$\sigma_d^2$ denotes the error weight caused by feedback delay in each
relay. The term $\frac{K+1}{2}$ scales the average delay for each relay when there is $K$ relay nodes. $\rho_\tau$ is the SNR of pilot signals in the training phase, $T_\tau$ is the duration of training phase, $B$ is the number of feedback bits for each relay, $\mathcal{J}_0$ is a Bessel function of the first kind of order 0, $f_D$ is the maximum Doppler shift, and $\tau$ is the feedback delay. Note that in (\ref{17}) and (\ref{18}), the calculations are approximations for analysis and numerical simulations. Substituting such $e_1^2$ and $e_2^2$ into the asymptotic capacities and divides the numerator
and the denominator by $K^2$, we see that the denominator will first decrease when $e_2^2$ is small and then increase when $e_2^2$ becomes large, which implies that the denominator will reach a
minimum value at some $K$. Therefore, there exists an optimal number
of relays to maximize the asymptotic capacities, which will be
confirmed by simulations.

\section{Simulation Results}
In this section, numerical results are carried out to validate what we draw from the analysis in the previous sections for the proposed beamforming design. We compare the robust MMSE-RZF with MF and MF-RZF in~\cite{10} and QR in~\cite{8}. The $\alpha^{\mathrm{RZF}}$ of MF-RZF in~\cite{10} is fixed to $1$. ZF mentioned in Section
\uppercase \expandafter {\romannumeral 3} is also plotted for reference. In all these figures, we set $M=N=4$, and focus on the performance of rates for various number of relays, SNR of BC and FC, and power of CSI errors. The ergodic rates are plotted by simulations through 1000 different channel realizations.
\subsection{Capacity Versus Number of Relays}
In Fig.~2, we compare the ergodic and asymptotic rate capacities
of the MMSE-RZF beamforming schemes for various regularizing
factors. We set PNR=QNR=$10dB$ and $e_1^2=e_2^2=0.01$. The curves are
the asymptotic rates, and the dots are the ergodic rates
obtained from simulation. The ergodic rate converges to the derived asymptotic rate for various $\alpha^{\mathrm{MMSE}}$ and $\alpha^{\mathrm{RZF}}$.
In Fig.~3, we compare the rate performance of the five
beamforming schemes.  The advantage of the proposed robust MMSE-RZF
can be observed. Note that MF, MF-RZF and ZF are all special cases
of MMSE-RZF, which are not optimized with the system condition. The bad
performance of QR can be explained as that its effective channel
matrix is not diagonal but upper triangular, which results in power
consumption. The  poor performance of ZF
 comes from the inverse Wishart distribution term in
its power control factor at the relays, especially when $M=N$.  We
find that the ergodic capacities still satisfy the scaling law
in~\cite{4}, i.e., $C=(M/2) \log(K)+O(1)$ for large $K$ in the
presence of CSI errors. This is also consistent with the  asymptotic
capacities derived for MMSE-RZF with QR SIC detection.
\begin{figure}
\centering
\includegraphics[width=3in]{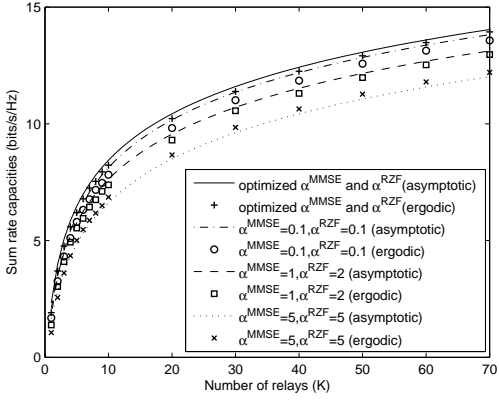}
\caption{Ergodic rates and asymptotic rates versus $K$ for various $\alpha^{\mathrm{MMSE}}$ and $\alpha^{\mathrm{RZF}}$.}
\end{figure}
\begin{figure}
\centering
\includegraphics[width=3in]{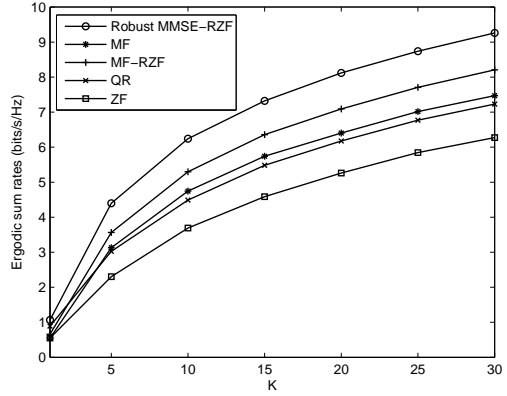}
\caption{Ergodic rates versus $K$. PNR=10dB, QNR=10dB. $e_1^2=e_2^2=-10dB$.}
\end{figure}

\subsection{Capacity Versus Power of CSI Error}
Fig.~4 compare the ergodic rate capacities versus the power of
 CSI error. we set $K=3$ and
PNR$=10dB$, QNR$=10dB$. In this case, MMSE-RZF also outperforms others as the power of CSI error increases. The superiority of the MMSE-RZF compared with MF and MF-RZF decreases as $e_1$ and $e_2$ increase, and ZF outperforms MF for small $e_1 (e_2)$ while underperforms MF for big $e_1 (e_2)$. This is because that when $e_1$ and $e_2$ are small, the optimal $\alpha^{\mathrm{MMSE}}$ and $\alpha^{\mathrm{RZF}}$ should be small, e.g., $\alpha^{\mathrm{MMSE,opt}}=0.5$ for PNR=$10dB$ and $e_1=0$. While in the MF and MF-RZF, $\alpha^{\mathrm{MMSE}}=\infty$ and $\alpha^{\mathrm{RZF}}=1$, which are far away from the optimal values. When $e_1$ and $e_2$ grows,  $\alpha^{\mathrm{MMSE,opt}}$ and $\alpha^{\mathrm{RZF,opt}}$ increases, which get close to the MF and MF-RZF.
\begin{figure}
\centering
\includegraphics[width=3in]{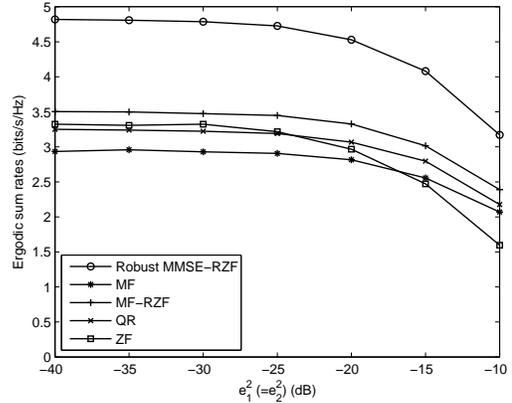}
\caption{Ergodic rates versus $e_1^2 (=e_2^2)$. PNR=10dB, QNR=10dB.}
\end{figure}

\subsection{Capacity Versus PNR and QNR}
In Fig.~5, we increase PNR and QNR
simultaneously for $K=5$. When  CSIs are perfect, the rates of all the five beamformings grow linearly
with the PNR (=QNR) in dB. When CSI
error occurs, we see the capacity limits. This is the ``ceiling effect" discussed in Section
\uppercase \expandafter {\romannumeral 5}. This can also be seen in (\ref{1}). If CSIs are perfect, SNR of each stream grows linearly with PNR (=QNR), so the rate grows linearly with PNR (=QNR) in dB. When CSI is imperfect, the numerator and denominator will simultaneously grows, resulting in a limit of SNR and the rate. The rate of ZF beamformer converges to the proposed robust MMSE-RZF beamformer at high SNR (PNR, QNR) for perfect CSI. This can be explained by the derived $\alpha^{\mathrm{MMSE,opt}}$ in (\ref{29}) and $\alpha^{\mathrm{RZF,opt}}$ in (\ref{24}), which both converge to zero at high SNR (large $P/\sigma_1^2$ and $Q/\sigma_2^2$) and perfect CSI ($e_1=e_2=0$). Since ZF is known to be the optimal beamforming for high SNR, the proposed MMSE-RZF is asymptotically optimal at high SNR with the QR SIC at the destination.
\begin{figure}
\centering
\includegraphics[width=3in]{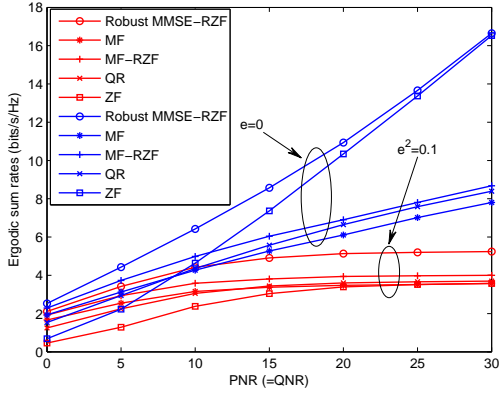}
\caption{Ergodic rates versus PNR (=QNR). $e_1^2=e_2^2=0$ or $e_1^2=e_2^2=-10dB$.}
\end{figure}

\subsection{Capacity Versus Relay Number for Dynamic CSI Error}
%Consider a practical case where the CSI erro varies with the number of relays. Let $e=\sigma_q+K\sigma_d$, where $\sigma_q$ denotes the quantization error due to the limit bits of feedback, and $\sigma_d$ denotes the error caused by delay in each relay. In

It is interesting to see Fig.~6 which shows the rates versus
the  relay number $K$ when CSI errors varies with $K$. This is a
very practical scenario when we assume that the error caused by
channel estimation is $\sigma_e^2=0.05$, $B=24$ for feedback,
$f_D=10Hz$ for $f=2.4GHz$ and $v=4.5Km/h$ for pedestrian speed
($f_D=vf/C$, $C$ denotes the speed of light), and $\tau=5ms$ which
is available in practical transmission. The feedback is based on the Lloyd VQ algorithm as in \cite{27}. It is observed that the
rates achieve maximum at some optimal relay number in the
presence of CSI error. For the assumption in this figure, the
optimal $K$ is $4$. For multi-relay networks, the processing and
feedback delay is always large. So we can save power and improve
performance by only choosing the optimal number of relays to forward the signal. The
selected relays can be constant or based on the instantaneous CSI.
\begin{figure}
\centering
\includegraphics[width=3in]{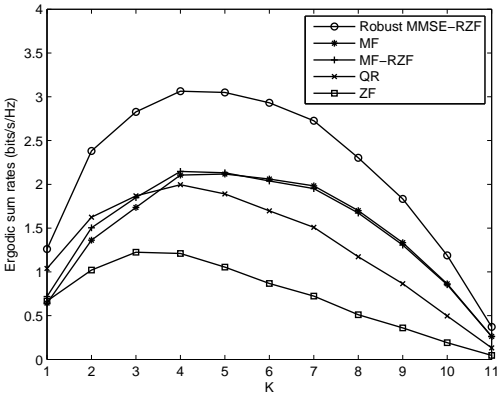}
\caption{Ergodic rates versus $K$. PNR=QNR=$10dB$. $e_1$ and $e_2^2$ are changed with number of relays $K$ as defined in (\ref{17}) and (\ref{18}).}
\end{figure}

\subsection{Capacity Versus Number of Feedback Bits}
In Fig.~7, using the same model as in Fig.~6, we focus on the effect of limited feedback. As can be observed, the rates increases fast with the number of feedback bits at the beginning, but slowly when the bits are enough to restore the CSIs. In Fig.~8, we further compare the performance under different number of feedback bits with perfect CSI case. It is observed that as the number of feedback bits increases, the rate approaches that of the perfect CSI case, and the rate is very close to that of perfect CSI case when the feedback bits are $12$.
\begin{figure}
\centering
\includegraphics[width=3in]{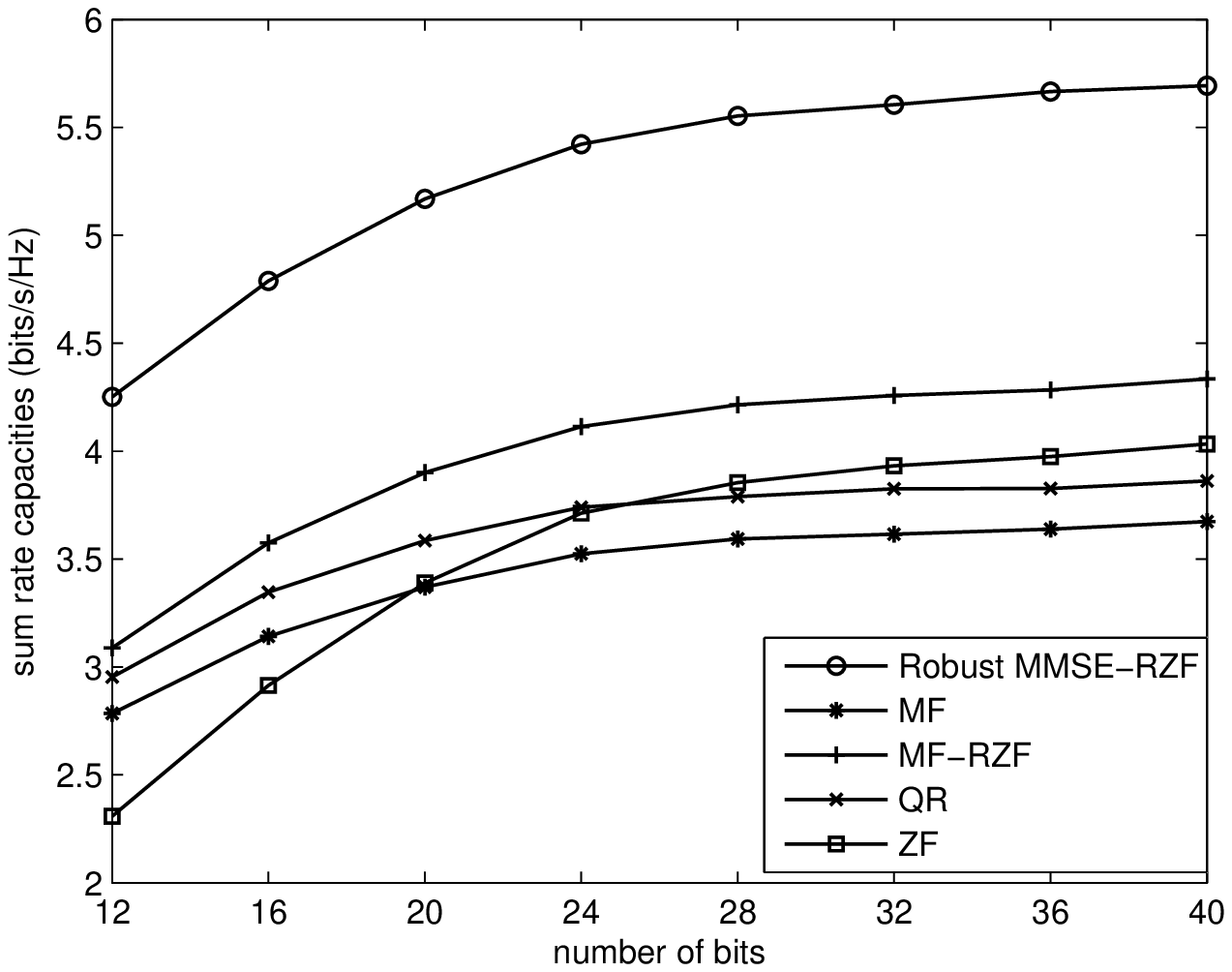}
\caption{Ergodic rates versus number of feedback bits. PNR=QNR=$10dB$. $e_1=0$ and $e_2^2$ is defined as in (\ref{18}).}
\end{figure}

\begin{figure}
\centering
\includegraphics[width=3in]{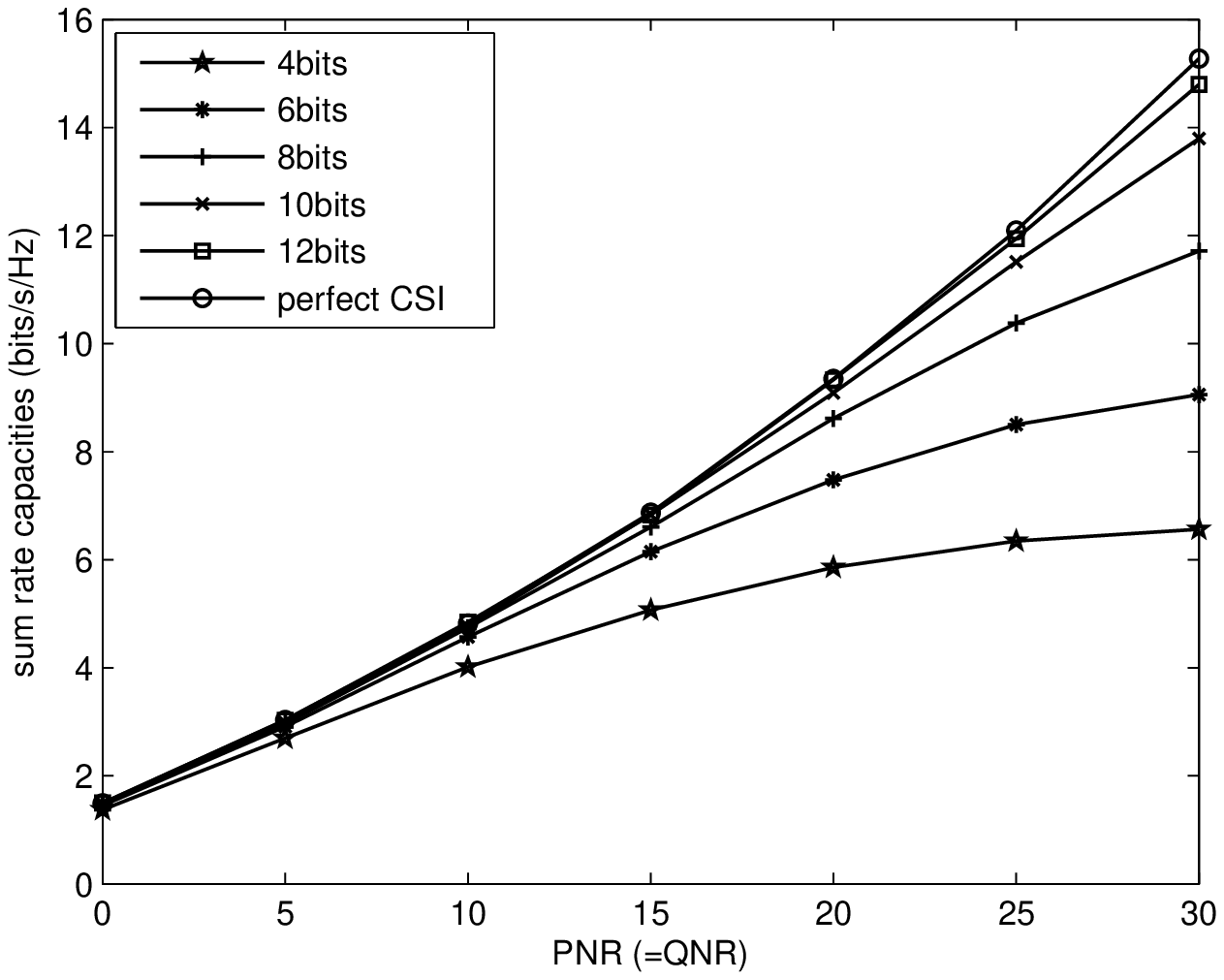}
\caption{Ergodic rates versus number of feedback bits. PNR=QNR=$10dB$. $e_1=0$ and $e_2^2$ is defined as in (\ref{18}).}
\end{figure}
%This
%property also lies in MF-RZF, but it is far beyond the number of
%relays in our simulation.

%\subsection{BER performance with CSI error}
%Another important criterion is the bit error rate (BER) performance
%of a system.  Here we suppose QPSK modulation. In Fig.~8, we shows
%the BER versus PNR (=QNR) for $K=3$ and $e=0,0.1$. When perfect CSI
%is available at relays, the BER will decrease to zero as the PNR
%(=QNR) goes to infinity. However, when CSI error occurs, the BER
%will become flat as the PNR (=QNR) goes to infinity. This is because
%that the powers of the effective signal and the CEG-noise will go to
%infinity at the same time as the PNR (=QNR) goes to infinity, and
%the BER will tend to a constant, which is the BER floor. Fig.~8 also
%shows that the MF-RZF is still the best choice when CSI error
%occurs.

\section{Conclusion}
In this paper, based on linear beamformer at the relay and QR SIC at
the destination,  we propose a robust MMSE-RZF beamformer optimized in terms of rate
 in a dual-hop MIMO multi-relay network with
Amplify-and-Forward (AF) relaying protocol in the presence of imperfect CSI. Since it is difficult to obtain the global optimal MMSE-RZF beamformer, we solve the optimization in two steps. The MMSE receiver is optimized by maximizing the SINR at relay nodes. The RZF precoder is optimized by maximizing the asymptotic rate derived upon a given MMSE receiver using Law of Large Number. Simulation results show that
the asymptotic rate matches with the ergodic rate.
Analysis and simulations demonstrate that the proposed robust
MMSE-RZF outperforms other coexistent beamforming schemes.

\appendix{}
We calculate square of the norm of the non-diagonal as (\ref{50}).
\begin{figure*}[!t]
\begin{equation}
\begin{split}
\left|({\bf{H}}_{{\mathcal{S}\mathcal {D}}})_{(i,j)}\right|^2
=&\left|\left(\sum_{k=1}^K\mathbf{Q}_k\frac{\mathbf{\Lambda}_k}{\mathbf{\Lambda}_k+\alpha^{\mathrm{RZF}}\mathbf{I}_M}\mathbf{Q}_k^H\mathbf{P}_k\frac{\mathbf{\Theta}_k}{\mathbf{\Theta}_k+\alpha^{\mathrm{MMSE}}\mathbf{I}_M}\mathbf{P}_k^H\right)_{(i,j)}\right|^2
\\=&\sum_k\left|\sum_{\ell,m,n}(\mathbf{Q}_k)_{i,\ell}\frac{\lambda_{k,\ell}}{\lambda_{k,\ell}+\alpha^{\mathrm{MMSE}}}(\mathbf{Q}_k)_{m,\ell}^*(\mathbf{P}_k)_{m,n}\frac{\theta_{k,n}}{\theta_{k,n}+\alpha^{\mathrm{RZF}}}(\mathbf{P}_k)_{j,n}^*\right|^2
\\=&\sum_{k,l,n,r,t}\sum_{m\neq i,j}\frac{\lambda_{k,\ell}}{\lambda_{k,\ell}+\alpha^{\mathrm{MMSE}}}\frac{\lambda_{k,r}}{\lambda_{k,r}+\alpha^{\mathrm{MMSE}}}\frac{\theta_{k,n}}{\theta_{k,n}+\alpha^{\mathrm{RZF}}}\frac{\theta_{k,t}}{\theta_{k,t}+\alpha^{\mathrm{RZF}}}
\\&(\mathbf{Q}_k)_{i,\ell}(\mathbf{Q}_k)_{m,\ell}^{*}(\mathbf{Q}_k)_{i,r}^{*}(\mathbf{Q}_k)_{m,r}
(\mathbf{P}_k)_{m,n}(\mathbf{P}_k)_{j,n}^{*}(\mathbf{P}_k)_{m,t}^{*}(\mathbf{P}_k)_{j,t}
\\+&\sum_{k,\ell,n,r,t}\frac{\lambda_{k,\ell}}{\lambda_{k,\ell}+\alpha^{\mathrm{MMSE}}}\frac{\lambda_{k,r}}{\lambda_{k,r}+\alpha^{\mathrm{MMSE}}}\frac{\theta_{k,n}}{\theta_{k,n}+\alpha^{\mathrm{RZF}}}\frac{\theta_{k,t}}{\theta_{k,t}+\alpha^{\mathrm{RZF}}}
\\&|(\mathbf{Q}_k)_{i,\ell}|^2|(\mathbf{Q}_k)_{i,r}|^2
(\mathbf{P}_k)_{i,n}(\mathbf{P}_k)_{j,n}^{*}(\mathbf{P}_k)_{i,t}^{*}(\mathbf{P}_k)_{j,t}
\\+&\sum_{k,\ell,n,r,t}\frac{\lambda_{k,\ell}}{\lambda_{k,\ell}+\alpha^{\mathrm{MMSE}}}\frac{\lambda_{k,r}}{\lambda_{k,r}+\alpha^{\mathrm{MMSE}}}\frac{\theta_{k,n}}{\theta_{k,n}+\alpha^{\mathrm{RZF}}}\frac{\theta_{k,t}}{\theta_{k,t}+\alpha^{\mathrm{RZF}}}
\\&(\mathbf{Q}_k)_{i,\ell}(\mathbf{Q}_k)_{j,\ell}^{*}(\mathbf{Q}_k)_{i,r}^{*}(\mathbf{Q}_k)_{j,r}
|(\mathbf{P}_k)_{j,n}|^2|(\mathbf{P}_k)_{j,t}|^2\label{50}
\end{split}
\end{equation}
\hrulefill
\begin{equation}
\begin{split}
\left|({\bf{H}}_{{\mathcal{S}\mathcal {D}}})_{(i,j)}\right|^2\overset{w.p.}{\longrightarrow}&\frac{K(M+2)}{(M+1)^2}
\mathrm{E}\left\{\frac{\lambda^2}{\left(\lambda+\alpha^{\mathrm{RZF}}\right)^2}\right\}
\mathrm{E}\left\{\frac{\theta^2}{\left(\theta+\alpha^{\mathrm{MMSE}}\right)^2}\right\}
\\-&\frac{K}{(M+1)^2}
\mathrm{E}\left\{\frac{\lambda\lambda^{\prime}}{\left(\lambda+\alpha^{\mathrm{RZF}}\right)\left(\lambda^{\prime}+\alpha^{\mathrm{RZF}}\right)}\right\}
\mathrm{E}\left\{\frac{\theta^2}{\left(\theta+\alpha^{\mathrm{MMSE}}\right)^2}\right\}
\\-&\frac{K}{(M+1)^2}
\mathrm{E}\left\{\frac{\lambda^2}{\left(\lambda+\alpha^{\mathrm{RZF}}\right)^2}\right\}
\mathrm{E}\left\{\frac{\theta\theta^{\prime}}{\left(\theta+\alpha^{\mathrm{MMSE}}\right)\left(\theta^{\prime}+\alpha^{\mathrm{MMSE}}\right)}\right\}
\\-&\frac{KM}{(M+1)^2}
\mathrm{E}\left\{\frac{\lambda\lambda^{\prime}}{\left(\lambda+\alpha^{\mathrm{RZF}}\right)\left(\lambda^{\prime}+\alpha^{\mathrm{RZF}}\right)}\right\}
\mathrm{E}\left\{\frac{\theta\theta^{\prime}}{\left(\theta+\alpha^{\mathrm{MMSE}}\right)\left(\theta^{\prime}+\alpha^{\mathrm{MMSE}}\right)}\right\}
\end{split}\label{51}
\end{equation}
\hrulefill
\end{figure*}
Since the fact in~\cite{RZF} that
\begin{eqnarray}
\mathrm{E}\left\{|(\mathbf{Q}_r)_{i,k}|^2|(\mathbf{Q}_r)_{\ell,k}|^2\right\}&=&\frac{2}{M(M+1)} \, \mathrm{if} \, i=\ell, \label{31}\\  \mathrm{E}\left\{|(\mathbf{Q}_r)_{i,k}|^2|(\mathbf{Q}_r)_{\ell,k}|^2\right\}&=&\frac{1}{M(M+1)}\, \mathrm{if} \, i\neq \ell,\label{32}
\end{eqnarray}
then when $i\neq m, \ell\neq r$ we have
\begin{equation}
\begin{split}
&\mathrm{E}\left\{(\mathbf{Q}_k)_{i,\ell}(\mathbf{Q}_k)_{m,\ell}^{*}(\mathbf{Q}_k)_{i,r}^{*}
(\mathbf{Q}_k)_{m,r}\right\}
\\=&\frac{1}{M(M-1)}\mathrm{E}\left\{\sum_{\ell,r=1,\ell\neq r}^M(\mathbf{Q}_k)_{i,\ell}(\mathbf{Q}_k)_{m,\ell}^{*}(\mathbf{Q}_k)_{i,r}^{*}
(\mathbf{Q}_k)_{m,r}\right\}
\\=&\frac{1}{M(M-1)}\mathrm{E}\left\{\sum_{\ell,r=1}^M(\mathbf{Q}_k)_{i,\ell}(\mathbf{Q}_k)_{m,\ell}^{*}(\mathbf{Q}_k)_{i,r}^{*}
(\mathbf{Q}_k)_{m,r}\right\}\\
&-\frac{1}{M(M-1)}\mathrm{E}\left\{\sum_{\ell=1}^M|(\mathbf{Q}_k)_{i,\ell}|^2|(\mathbf{Q}_k)_{m,\ell}|^2\right\}
\\=&\frac{\mathrm{E}\left\{\left(\sum_{\ell=1}^M(\mathbf{Q}_k)_{i,\ell}(\mathbf{Q}_k)_{m,\ell}^{*}\right)\left(\sum_{r=1}^M(\mathbf{Q}_k)_{i,r}(\mathbf{Q}_k)_{m,r}^{*}\right)\right\}}{M(M-1)}\\
&-\frac{1}{M(M-1)}\cdot\frac{1}{(M+1)}
\\=&-\frac{1}{(M-1)M(M+1)}.\label{33}
\end{split}
\end{equation}
Substituting (\ref{31}), (\ref{32}) and (\ref{33}) into (\ref{50}), and through some manipulation,  we have (\ref{51}),
where $\lambda$ and $\lambda^{\prime}$, $\theta$ and $\theta^{\prime}$ are different singular values within one decomposition. Obviously, (\ref{51}) equals zero when $\alpha^{\mathrm{MMSE}}$ and $\alpha^{\mathrm{RZF}}$ are zero.

\end{document}